\documentclass[journal]{IEEEtran}
\usepackage[table]{xcolor}
\usepackage{booktabs}
\usepackage{multirow}
\usepackage{color}
\usepackage{subcaption}
\usepackage{algorithm}
\usepackage{algorithmicx}
\usepackage{algpseudocode}
\usepackage{amsfonts}
\usepackage{amsmath}
\usepackage{amsthm}
\usepackage[utf8]{inputenc}
\usepackage[english]{babel}
\usepackage{cleveref}
\makeatletter
\renewcommand{\@IEEEsectpunct}{}
\makeatother

\DeclareMathOperator*{\argminB}{argmin}   

\newtheorem{definition}{\bf Definition}
\def\BibTeX{{\rm B\kern-.05em{\sc i\kern-.025em b}\kern-.08em
    T\kern-.1667em\lower.7ex\hbox{E}\kern-.125emX}}
\newtheorem{theorem}{Theorem}

\ifCLASSINFOpdf
   \usepackage[pdftex]{graphicx}
\else
\fi
\begin{document}
%
\title{Utility-aware Privacy-preserving Data Releasing}
%
%
%

\author{Di~Zhuang,~\IEEEmembership{Student~Member,~IEEE,}
        and~J.~Morris~Chang,~\IEEEmembership{Senior~Member,~IEEE}
\thanks{The authors are with the Department of Electrical Engineering, University of South Florida, Tampa, FL 33620 USA (e-mail: dizhuang@mail.usf.edu; chang5@usf.edu).}
 }

\maketitle

\begin{abstract}
In the big data era, more and more cloud-based data-driven applications are developed that leverage individual data to provide certain valuable services (the utilities). On the other hand, since the same set of individual data could be utilized to infer the individual's certain sensitive information, it creates new channels to snoop the individual's privacy. Hence it is of great importance to develop techniques that enable the data owners to release privatized data, that can still be utilized for certain premised intended purpose. Existing data releasing approaches, however, are either privacy-emphasized (no consideration on utility) or utility-driven (no guarantees on privacy). In this work, we propose a two-step perturbation-based utility-aware privacy-preserving data releasing framework. First, certain predefined privacy and utility problems are learned from the public domain data (background knowledge). Later, our approach leverages the learned knowledge to precisely perturb the data owners' data into privatized data that can be successfully utilized for certain intended purpose (learning to succeed), without jeopardizing certain predefined privacy (training to fail). Extensive experiments have been conducted on Human Activity Recognition, Census Income and Bank Marketing datasets to demonstrate the effectiveness and practicality of our framework.
\end{abstract}


%
\IEEEpeerreviewmaketitle

\section{Introduction} \label{sec:Introduction}
\IEEEPARstart{A}{s} the advent and advance of cloud computing and data science in this big data era, more and more cloud-based data-driven applications are developed by different service providers (the data users, such as Facebook, LinkedIn and Google). Most of these applications leverage the vast amount of data collected from each individual (the data owner) to offer certain valuable service back to the corresponding individual or for the other political and commercial purposes, such as friend recommendation, human activity recognition, health monitoring, targeted advertising and election prediction. However, the same set of data could be repurposed in different ways to infer certain sensitive personal information, which would jeopardize the individual's privacy.



In the recent Facebook data leak scandal (April, 2018) \cite{FacebookCambridge2018}, about 87 million Facebook users' data were collected by a Facebook quiz app (a cloud-based data-driven application) and then paired with information taken from their social media profile (including their gender, age, relationship status, location and ``likes'') without any privacy-preserving operations being taken other than anonymization. Thus, the data user or the other adversaries that have the access to the data can still infer certain sensitive information of each individual from his/her data, such as identity, sexual orientation and marital status.
The unprecedented data leak scandal raised the alarm of privacy concerns among cloud-based data-driven applications which could became a big obstacle that impedes the individuals from releasing their data to the service providers to receive valuable service (the utilities).

A similar situation could happen in the patient-hospital scenario as shown in Fig.~\ref{fig:example}. Patient Alice (the data owner) would like to release her data to hospital Bob (the data user) with the premise of using it for disease A diagnosis. However, people like Eve (could be Bob), who works in the same hospital and has the access to Alice's data, could use the same data to infer certain irrelevant sensitive information about Alice, such as her disease B diagnosis.
In this case, some individuals (e.g., Facebook users or Alice) would like to release their data to receive good utilities, while on the premise that the service providers are prevented from inferring certain sensitive information from their data (e.g., identity, sexual orientation and marital status).
Therefore, it is of vital importance to develop a utility-aware privacy-preserving data releasing framework for cloud-based data-driven applications, which enables the released data to be utilized for certain premised intended purpose (utility target), without jeopardizing the corresponding data owner's certain privacy target.

\begin{figure}[tb]
\centering
\includegraphics[width=240pt]{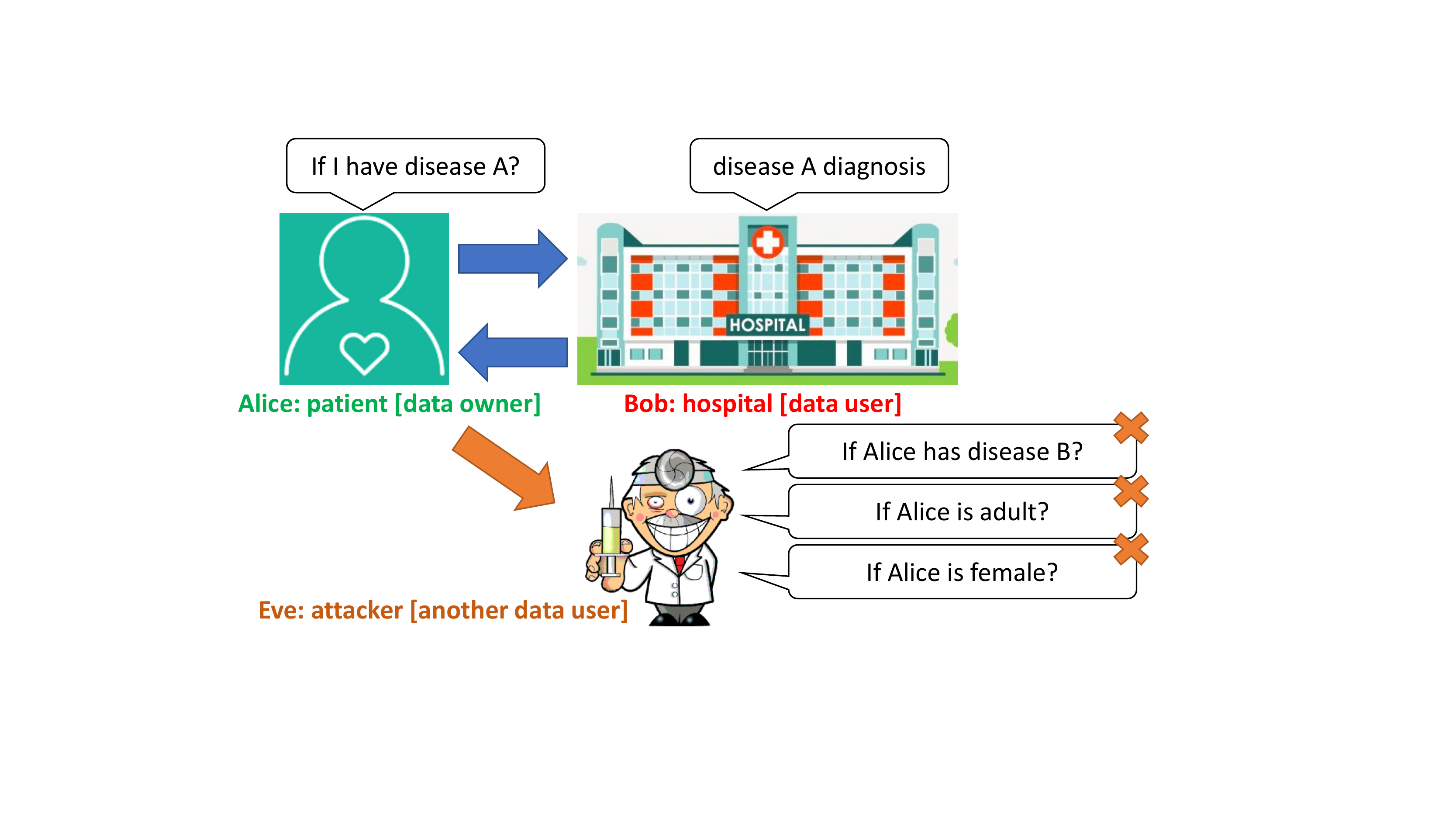}
\caption{An example in patient-hospital scenario.}
\label{fig:example}
\end{figure}

Designing such general utility-aware privacy-preserving data releasing framework is rather challenging. To date, a few related approaches have been proposed \cite{sweeney2002k, kim2003multiplicative, zhang2012functional, shokri2015privacy, abadi2016deep, kung2015discriminant, diamantaras2016data, kung2017compressive, zhuang2017fripal}. However, these approaches cannot fulfil all the privacy requirements needed in the cloud-based data-driven application scenario. For example, approaches that relied on additive or multiplicative random noise perturbation \cite{kim2003multiplicative} and k-anonymity \cite{sweeney2002k} cannot handle the curse of dimensionality.
Differential privacy machine learning approaches \cite{zhang2012functional, shokri2015privacy, abadi2016deep} have been proposed to publish machine learning models while preserving the training data privacy.
In this paper, however, we consider the scenario that the machine learning models have been trained in advance by the cloud-based service providers (the data users). The data to be protected would appear as the testing data, which is beyond the scope of those approaches. Besides, \cite{hitaj2017deep} has shown that some record-level differential privacy approaches applied to collaborative learning scenario are ineffective in dealing with inference attacks. Dimensionality reduction based approaches \cite{kung2015discriminant, diamantaras2016data, kung2017compressive, zhuang2017fripal, al2017ratio} have been proposed to preserve privacy while maintaining most of the utility. However, despite of their good experimental performance on several public datasets, those approaches didn't introduce any uncertainty to hide the sensitive information, which failed to show the needed guarantees on the privacy targets mathematically.

To address the challenges mentioned above, in this paper, we devote to solve a two-party exemplar problem. The data user (i.e., the cloud-based service provider) use his/her domain knowledge and public domain data to train a model to provide certain service in advance. The data owner would like to receive the service via sharing his/her own data as the testing data to the data user. The data owner predefines several privacy targets (sensitive information) that he/she would like to prevent the data user from inferring from his/her data. By ``predefines'', we assume that the data owner knows what the malicious inference and the corresponding domain knowledge and public domain data will be utilized by the malicious data users.

In this paper, a two-step perturbation-based utility-aware privacy-preserving data releasing framework is proposed to tackle this problem.
Given certain specific utility/privacy targets (i.e., the inference problems and the corresponding domain knowledge and public domain data), our approach precisely transforms the original data into privatized data that can be successfully utilized for certain intended purpose (learning to succeed), without jeopardizing certain predefined privacy (training to fail).
The first step is a coarse-grained data perturbation method, Joint Utility/Privacy Analysis (JUPA). JUPA is an subspace-optimized projection method, which combines the advantages from both DCA \cite{kung2015discriminant} (utility driven projection) and MDR \cite{diamantaras2016data} (privacy emphasized projection), and tries to find a subspace projection that could optimize for both utility and privacy targets with the knowledge learned from the public datasets.
The second step is a fine-grained data perturbation method inspired by the ``label changing'' problems (e.g., adversarial image perturbation \cite{szegedy2013intriguing, Goodfellow2014, gardner2015deep, papernot2016limitations, papernot2017practical, carlini2017towards, athalye2017synthesizing}) in the computer vision area, where in order to change the image's class membership, very small perturbations are added to the image that remain quasi-imperceptible to a human vision system. For instance, \cite{gardner2015deep} proposed a Maximum Mean Discrepancy \cite{fortet1953convergence} (MMD) statistic test related approach to make semantic change to the appearance of given images.
We propose to use a MMD-like loss function to leverage the knowledge (i.e., the discriminant distance among the classes in each privacy target) learned from the public domain dataset and precisely perturb each coarse-grain-perturbed data to a fine-grain-perturbed data that belongs to a randomly selected privacy target class (the data owner's secret parameter).

In the experiments, we have tested our frame on three public datasets: Human Activity Recognition, Census Income and Bank Marketing datasets. The experiment results demonstrate that (a) JUPA is a more general utility-aware dimensionality reduction approach compared with DCA \cite{kung2015discriminant} and MDR \cite{diamantaras2016data}; (b) given certain predefined privacy target, our fine-grained data perturbation approach can reduce the accuracy of the corresponding inference attack to the level of random guessing.

The rest of paper is organized as follows:
Section~\ref{sec:RelatedWork} presents the related works.
Section~\ref{sec:Preliminaries} presents the preliminaries about dimensionality reduction and maximum mean discrepancy.
Section~\ref{sec:Framework} describes our proposed utility-aware privacy-preserving data releasing framework.
Section~\ref{sec:ExperimentalEvaluation}  presents the experimental evaluation.
Section~\ref{sec:Conclusion} presents the conclusion and future work.

\section{Related Work} \label{sec:RelatedWork}
A few privacy-preserving data releasing approaches have been proposed, including solutions based on cryptography \cite{erkin2012generating, nikolaenko2013privacy, bost2015machine, bonawitz2017practical}, differentially private synthetic data generation \cite{hardt2012simple, jiang2013differential, bindschaedler2017plausible, soria2017differentially}, and dimensionality reduction \cite{liu2006random, kung2015discriminant, diamantaras2016data, kung2017compressive, zhuang2017fripal, al2017ratio}. Most of the cryptography-based approaches are designed for specific applications/algorithms. For instance, \cite{nikolaenko2013privacy} developed a privacy-preserving ridge regression system that utilized additive homomorphic encryption and Garbled circuits to train a ridge regression model with the encrypted data statistic shares submitted by multiple data owners. \cite{bost2015machine} proposed to use cryptographic building blocks to enable testing new samples while protecting both the ML model and the submitted samples, in three popular classification protocols: hyperplane decision, Naïve Bayes, and decision trees. Although cryptography-based approaches prevent the adversaries from performing inference attack on the encrypted data/model, they are not flexible enough to work for general data releasing purpose.

Differential privacy (DP) \cite{dwork2008differential} is one of the most popular standard for quantifying individual privacy. DP aims to protect the privacy of individuals via adding randomness to the aggregate information. Differentially private synthetic data generation approaches utilize those differentially private aggregate information to generate synthetic data. For instance, \cite{jiang2013differential} considers to use differential privacy component analysis for data releasing. ``Plausible Deniability'' \cite{bindschaedler2017plausible}, has been proposed and achieved by applying a privacy test after generating the synthetic data. The generative model proposed in \cite{bindschaedler2017plausible} is a probabilistic model which captures the joint distribution of features based on correlation-based feature selection (CFS) \cite{hall1999correlation}. \cite{hardt2012simple} proposed an algorithm which combines the multiplicative weights approach and exponential mechanism for differentially private data release. \cite{soria2017differentially} proposed a micro-aggregation \cite{domingo2005ordinal} based differential private data releasing approach which reduces the noise required by differential privacy based on $k$-anonymity. Although DP-based approaches provide strong guarantees on individuals' privacy, they does not take any utility targets into account in designing their privacy-preserving data releasing mechanisms.

The dimensionality reduction approaches provide a promising way to irreversibly transform the original data, and publish the transformed data for general usage. \cite{liu2006random} proposed to use random projection matrix to project the original data to a lower dimensional space. However, the random projection method mainly focuses on the privacy targets without considering the utility targets, which downgrades its utility performance. A few dimensionality reduction based privacy-preserving approaches focusing on maintaining the utility have been proposed \cite{kung2015discriminant, diamantaras2016data, kung2017compressive, zhuang2017fripal, al2017ratio}. For instance, \cite{kung2015discriminant} proposed to use Discriminant Component Analysis (DCA), a supervised version of Principle Component Analysis (PCA), to project the data into a lower dimensional space that maximizes the discriminant power for specific targets. However, since DCA mainly focuses on the utility target, it might maintain the utility while somewhat preserve the privacy because of the information loss through the dimensionality reduction. However, DCA could not control or adjust the projection matrix in terms of the privacy target. \cite{diamantaras2016data} proposed Multi-class Discriminant Ratio (MDR), which projects the data based on a pair of classification targets, (a) a privacy-insensitive and (b) a privacy-sensitive target. RUCA \cite{al2017ratio}, improves the MDR to provide more flexibility to adjust the trade-off or tuning between utility and privacy. However. these approaches do not introduce any uncertainty/randomness to hide the sensitive information, which failed to show the needed guarantees on the privacy targets mathematically.

\section{Preliminaries} \label{sec:Preliminaries}
\subsection{Dimensionality Reduction via Eigenvalue Decomposition} \label{sec:SPED}
An important component of our framework is supervised dimensionality reduction technique (i.e., it relies on data labels). Consider a dataset with $N$ training samples $\{x_{1}, x_{2}, \dots, x_{N}\}$, where each sample has $M$ features ($x_{i} \in \mathbb{R}^{M}$). Since the same dataset could be utilized in different classification problems, each classification problem $c$ has a unique set of labels $L^{c}_{i}$ associated with the corresponding training samples $x_{i}$. Without loss of generality, we assume the dataset could be utilized for a single utility target $U$ and a single privacy target $P$. Then, each training sample $x_{i}$ has two labels $L^{u}_{i} \in \{1, 2, \dots, L^{u}\}$ and $L^{p}_{i} \in \{1, 2, \dots, L^{p}\}$. $L^{u}$ and $L^{p}$ are the numbers of classes of the utility target and the privacy target, respectively.

Based on Fisher's linear discriminant analysis \cite{fisher1936use, mika1999fisher}, given a classification problem, the within-class scatter matrix of its training samples contains most of the ``noise information'', while the between-class scatter matrix of its training samples contains most of the ``signal information''.

We define the within-class scatter matrix and the between-class scatter matrix for the utility target as follows:
\begin{equation} \label{eq:S_U1}
            S_{W_{U}} = \sum_{l=1}^{L^{u}} \bigg( \sum_{i=1}^{N_{l}^{u}} x_{i}x_{i}^{T}- N_{l}^{u} \mu_{l}\mu_{l}^{T} \bigg)
\end{equation}
\begin{equation} \label{eq:S_U2}
            S_{B_{U}} = \sum_{l=1}^{L^{u}} N_{l}^{u}\mu_{l}\mu_{l}^{T} - N\mu\mu^{T}
\end{equation}
where $\mu=\frac{1}{N}\sum_{i=1}^{N} x_{i}$, $\mu_{l}$ is the mean vector of all training samples belonging to class $l$, $N_{l}^{u}$ is the number of training samples belonging to class $l$ of the utility target.

Similarly, for the privacy target the within-class scatter matrix and the between-class scatter matrix define as:
\begin{equation} \label{eq:S_P1}
    \begin{split}
            S_{W_{P}} = \sum_{l=1}^{L^{p}} \bigg( \sum_{i=1}^{N_{l}^{p}} x_{i}x_{i}^{T}- N_{l}^{p} \mu_{l}\mu_{l}^{T} \bigg)
    \end{split}
\end{equation}
\begin{equation} \label{eq:S_P2}
    \begin{split}
            S_{B_{P}} &= \sum_{l=1}^{L^{p}} N_{l}^{p}\mu_{l}\mu_{l}^{T} - N\mu\mu^{T}
    \end{split}
\end{equation}

Let $W$ be an $K \times M$ projection matrix, in which $K < M$. Given testing sample $x$, $\hat{x}=x^{T} \cdot W$ is its subspace projection. Our framework combines the
advantages of two eigenvalue decomposition based dimensionality reduction techniques: DCA \cite{kung2015discriminant} (utility driven projection) and MDR \cite{diamantaras2016data} (privacy emphasized projection).

\subsubsection{Discriminant Component Analysis (DCA)} \label{sec:DCA}
DCA \cite{kung2015discriminant} involves searching for the projection matrix $W \in R^{M \times K}$:
\begin{equation} \label{eq:DCA_1}
    \begin{split}
            DCA=\frac{det(W^{T} S_{B_{U}} W)}{det(W^{T} (\bar{S} + \rho I) W)}
    \end{split}
\end{equation}
where $det(\cdot)$ is the determinant operator, $\rho I$ is a small regularization term added for numerical stability, and $\bar{S}=S_{W_{U}}+S_{B_{U}}= \sum_{i=1}^{N} x_{i}x_{i}^{T} - N\mu\mu^{T}$.

The optimal solution to this problem can be derived from the first $K$ principal generalized eigenvectors of the matrix pencil $(S_{B_{U}}, \bar{S} + \rho I)$.
\subsubsection{Multi-class Discriminant Ratio (MDR)} \label{sec:MDR}
MDR \cite{diamantaras2016data} considers both the utility target and the privacy target, which is defined as:
\begin{equation} \label{eq:MDR}
    \begin{split}
            MDR=\frac{det(W^{T} (S_{B_{U}}) W)}{det(W^{T} (S_{B_{P}} + \rho I) W)}
    \end{split}
\end{equation}
where $\rho I$ is a small regularization term added for numerical stability.

The optimal solution to MDR can be derived from the first $K$ principal generalized eigenvectors of the matrix pencil $(S_{B_{U}}, S_{B_{P}} + \rho I)$.

\subsection{Maximum Mean Discrepancy (MMD)} \label{sec:MMD}
The Maximum Mean Discrepancy \cite{fortet1953convergence} (MMD) statistic has been proposed to test whether two distributions $p$ and $q$ are different based on the samples drawn from each of them. In this work, our fine-grained data perturbation utilized a MMD-like loss function inspired by a kernel-MMD solution \cite{gretton2007kernel}. Let $p$ and $q$ be two distributions defined on a domain $\mathcal{X}$. Given observations $X := \{x_1, x_2, \dots, x_m\}$ and $Y := \{y_1, y_2, \dots, y_n\}$, drawn i.i.d. from $p$ and $q$ respectively, the kernel-MMD solution \cite{gretton2007kernel} is defined as:
\begin{equation} \label{eq:KMMD}
    \begin{split}
        MMD[\mathcal{F}, X, Y] =& \frac{1}{m}\sum_{i=1}^{m} \phi(x_{i}) - \frac{1}{n}\sum_{i=1}^{n} \phi(y_{i}) \\
        =&\Big[\frac{1}{m^{2}}\sum_{i, j=1}^{m} k(x_i, x_j) - \frac{2}{mn}\sum_{i, j=1}^{m, n} k(x_i, y_j) \\
        &+ \frac{1}{n^{2}}\sum_{i, j=1}^{n} k(y_i, y_j)\Big]^{\frac{1}{2}}
    \end{split}
\end{equation}
where $\mathcal{F}$ is a unit ball in a universal RKHS $\mathcal{H}$, defined on the compact metric space $\mathcal{X}$, with associated kernel $k(\cdot, \cdot)$, and $\phi(x)=k(x, \cdot)$. $MMD[\mathcal{F}, X, Y] \approx 0$, if and only if $p=q$.

\begin{figure*}[tb]
\centering
\includegraphics[width=525pt]{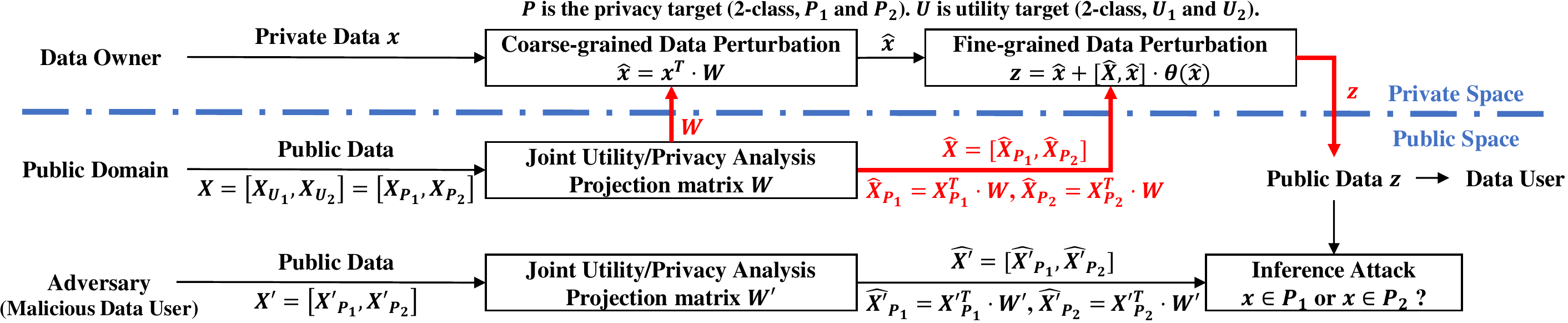}
\caption{A Utility-aware Privacy-preserving Data Releasing Framework.}
\label{fig:sysview}
\end{figure*}

\section{Utility-aware Privacy-preserving Data Releasing Framework} \label{sec:Framework}
\subsection{Framework Overview} \label{sec:ProblemStatement}
{\bf \textit{Problem Statement.}} As illustrated in Fig.~\ref{fig:sysview}, our framework involves two parties: the data owner(s) and the data user(s). The data user uses public data (background knowledge) to train a machine learning model (i.e., classification model) in advance to provide certain service (the utility targets). The data owner would like to release her private data to the data user for the purpose of the utility targets, and prevent the malicious data user from inferring certain predefined sensitive information (the privacy targets). Assume the data owner and the data user have access to similar set of public data (background knowledge) utilized for both utility and privacy targets, but the data owner does't know the data user's machine learning model. Our goal is to perturb the data owner's private data $x$ into perturbed data $z$ with the knowledge of predefined utility and privacy targets, such that the perturbed data $z$ could be utilized successfully for the intended purposes (i.e., the utility target achieves similar accuracies using either $x$ or $z$), without jeopardizing the data owner's privacy (i.e., the privacy target get no better accuracy than random guessing while using $z$). To achieve this goal, we propose a two-step data perturbation framework (Fig.~\ref{fig:sysview}).
More details about the two steps are described in Section~\ref{sec:CoarsePerturbation} and Section~\ref{sec:FinePerturbation}.

{\bf \textit{Threat Model.}} The adversaries in our framework are the malicious data users, who have the access to the public data that could be utilized as the training data for certain predefined privacy target. The adversaries would like to infer the knowledge (i.e., class) of the privacy target (i.e., classification problems) associated with the data owner's private data based on the corresponding perturbed data and public data (background knowledge).
For instance, as shown in Fig.~\ref{fig:sysview}, we shall assume that the predefined privacy target is a two-class (i.e., $\{P_1, P_2\}$) classification problem (utility targets are independent from the privacy task). Let $X=[X_{P_1}, X_{P_2}]$ be the public training samples for the privacy target, where $X_{P_i}$ ($i \in \{1, 2\}$) presents the samples associated with class $i$. Let $x$ be the data owner's private data, where $s \in \{P_1, P_2\}$ is its original (privacy target) class and $t \in \{P_1, P_2\}$ is its expected (privacy target) class after our privacy-preserving data releasing operation. The data owner could publish $z$ (i.e., the perturbed version of $x$) using our framework $F$: $z \leftarrow F(x, t, X_{P_1}, X_{P_2})$. The adversary has to use his/her approach $A(z, X_{P_1}, X_{P_2})$ to guess/infer the original (privacy task) class $s$.


\subsection{Coarse-grained Data Perturbation} \label{sec:CoarsePerturbation}
In this section, we introduce a general dimensionality reduction method Joint Utility/Privacy Analysis (JUPA). JUPA combines the advantages from both DCA \cite{kung2015discriminant} (utility driven projection) and MDR \cite{diamantaras2016data} (privacy emphasized projection), and tries to find a subspace projection that could optimize for both utility and privacy targets with the knowledge learned from the public datasets. Our problem settings are exactly the same as described in Section~\ref{sec:SPED}. For simplicity, we shall start from a single utility/privacy target scenario. JUPA tries to find a projection matrix $W$ that maximize the following function:
\begin{equation} \label{eq:JUPA}
    \begin{split}
            JUPA=\frac{det(W^{T} (S_{B_{U}} + \rho^{\prime}_{1}S_{W_{P}}) W)}{det(W^{T} (S_{W_{U}}+\rho_{1}S_{B_{P}} + \rho_{0}I) W)}
    \end{split}
\end{equation}
where $det(\cdot)$ is the determinant operator, $\rho_{0}$ is regularization parameter added for numerical stability, and $\rho_{1}$, $\rho^{\prime}_{1}$ are privacy-utility adjustment parameters.

The optimal solution to JUPA can be derived from the first $K$ principal generalized eigenvectors of the matrix pencil $(S_{B_{U}} + \rho^{\prime}_{1}S_{W_{P}}, S_{W_{U}}+\rho_{1}S_{B_{P}} + \rho_{0}I)$. After getting the projection matrix $W$, we perturb the data owner's private data $x$ and the training data matrix $X$ as $\hat{x}=x^{T} W$ and $\hat{X}=X^{T} W$.

Additionally, JUPA can be generalized to multiple utility/privacy targets by including multiple corresponding scatter matrices:
\begin{equation} \label{eq:JUPA_multiple}
    \begin{split}
            JUPA=\frac{ det(W^{T}( \sum_{i=1}^{N_u}S_{B_{U_{i}}} + \sum_{i=1}^{N_p}\rho^{\prime}_{i}S_{W_{P_{i}}}) W)}{det(W^{T} (\sum_{i=1}^{N_u}S_{W_{U_{i}}}+\sum_{i=1}^{N_p}\rho_{i}S_{B_{P_{i}}} + \rho_{0}I) W)}
    \end{split}
\end{equation}

{\bf \textit{Utility vs. ``Somewhat Privacy''.}} JUPA maintains a trade-off between the utility and ``somewhat privacy''. ``somewhat privacy'' means our coarse-grained perturbation approach optimizes towards privacy, but could not provide privacy guarantee (as in Section~\ref{sec:FinePerturbation}). On one hand, JUPA optimizes a subspace projection that maximizes the ``signal to noise'' ratio of the utility targets. On the other hand, JUPA optimizes towards two ``mappings'' for privacy targets: a ``many-to-one'' mapping, after which data belonging to the same privacy class are near each other (tuned by $\rho^{\prime}_{1}$); and a ``one-to-many'' mapping, after which data belonging to different privacy classes are far from each other (tuned by $\rho_{1}$). By adjusting $\rho_{1}$ and $\rho^{\prime}_{1}$, JUPA could be tuned between DCA \cite{kung2015discriminant}, MDR \cite{diamantaras2016data} and RUCA \cite{al2017ratio}. For instance, if $\rho_{1}=\rho^{\prime}_{1}=0$, this projection method becomes DCA; if $\rho_{1}$ is very large and $\rho^{\prime}_{1}=0$, it becomes MDR as the term $S_{B_{P}}$ dominates $(S_{W_{U}}+\rho_{1}S_{B_{P}} + \rho_{0}I)$; and if $\rho^{\prime}_{1}=0$ it becomes RUCA. Higher value of $\rho_{1}$ and $\rho^{\prime}_{1}$ means more emphasis on the privacy targets.

\subsection{Fine-grained Data Perturbation} \label{sec:FinePerturbation}
In this section, we introduce a perturbation approach that gradually change the privacy target classification label of a given data owner's coarse-grained perturbed data $\hat{x}$ from its source (original) label $s$ to a randomly selected target label $t$, via adding precisely calculated noise. For simplicity, we shall assume a single 2-class ($\{P_1, P_2\}$) privacy target scenario. Except for the data owner's coarse-grained perturbed data $\hat{x}$, another input for this approach is the coarse-grained perturbed training data matrix $\hat{X}=[\hat{X}_{P_{1}}, \hat{X}_{P_{2}}]$, where $\hat{X}$ will be split into two parts: $\hat{X}^{G}=[\hat{X}^{G}_{P_{1}}, \hat{X}^{G}_{P_{2}}]$ and $\hat{X}^{V}=[\hat{X}^{V}_{P_{1}}, \hat{X}^{V}_{P_{2}}]$. $\hat{X}^{G}$ is the ``ground truth'' training data matrix being used to gradually ``train'' the fine-grained perturbed data. $\hat{X}^{V}$ is the ``verification'' training data matrix being used to verify the current label of the input data $\hat{x}$ and intermediate perturbed data.

Given coarse-perturbed private data $\hat{x}$, we start from randomly selecting a target label $t \in \{P_1, P_2\}$ for $\hat{x}$, and use the following function to decide its current (source) label $s \in \{P_1, P_2\}$:
\begin{equation} \label{eq:sanitization}
    \begin{split}
        s=&label(\hat{x})=\argminB_{\{l : l \in \{P_{1}, P_{2}\}\}} \big(\frac{1}{|\hat{X}^{V}_{l}|}\sum_{\hat{x}_{i} \in \hat{X}^{V}_{l}} \phi(\hat{x}_{i}) - \phi(\hat{x}) \big)^{2}\\
        =&\argminB_{\{l : l \in \{P_{1}, P_{2}\}\}} \frac{1}{|\hat{X}^{V}_{l}|^{2}}\sum_{\hat{x}_{i}, \hat{x}_{j} \in \hat{X}^{V}_{l}} k(\hat{x}_{i}, \hat{x}_{j}) \\
        &- \frac{2}{|\hat{X}^{V}_{l}|}\sum_{\hat{x}_{i} \in \hat{X}^{V}_{l}} k(\hat{x}_{i}, \hat{x}) + k(\hat{x}, \hat{x}) \\
        =&\argminB_{\{l : l \in \{P_{1}, P_{2}\}\}} \frac{1}{|\hat{X}^{V}_{l}|^{2}}\sum_{\hat{x}_{i}, \hat{x}_{j} \in \hat{X}^{V}_{l}} k(\hat{x}_{i}, \hat{x}_{j}) \\
        &- \frac{2}{|\hat{X}^{V}_{l}|}\sum_{\hat{x}_{i} \in \hat{X}^{V}_{l}} k(\hat{x}_{i}, \hat{x})
    \end{split}
\end{equation}

Our approach perturbs $\hat{x}$ in an iterative fashion. Let $z_{i}$ be the $i$th ($i=1, 2, \dots$) intermediate perturbed data. Then, our iterative data sanitization function is defined as:
\begin{equation} \label{eq:sanitization}
    \begin{split}
        z_{0}=&\hat{x} \\
        z_{i}=&z_{i-1}+\theta(z_{i-1}) \ \ (i=1, \ 2, \ \dots)
    \end{split}
\end{equation}
where $\theta(z_{i})$ is the noise vector being added to $z_{i}$. The starting noise vector $\theta(z_{0})$ could be initiated as a zero vector or a random vector.

In order to compute $\theta(z_{i})$ in each iteration, inspired by the kernel-MMD solution \cite{gretton2007kernel} described in Section~\ref{sec:MMD}, we define a loss function as:
\begin{equation} \label{eq:loss}
    \begin{split}
        L(\theta(z_{i}))=&\big(\frac{1}{n_t}\sum_{\hat{x}_{i} \in \hat{X}^{G}_{t}} \phi(\hat{x}_{i}) - \phi(z_{i}) \big)^{2} \\
        & - \big(\frac{1}{n_s}\sum_{\hat{x}_{i} \in \hat{X}^{G}_{s}} \phi(\hat{x}_{i}) - \phi(z_{i}) \big)^{2} + \frac{\lambda}{2} \|\theta(z_{i})\|_{2}^{2} \\
        =&\frac{1}{n_t^{2}}\sum_{\hat{x}_{i}, \hat{x}_{j} \in \hat{X}^{G}_{t}} k(\hat{x}_{i}, \hat{x}_{j}) - \frac{1}{n_s^{2}}\sum_{\hat{x}_{i}, \hat{x}_{j} \in \hat{X}^{G}_{s}} k(\hat{x}_{i}, \hat{x}_{j}) \\
        &+ \frac{2}{n_s}\sum_{\hat{x}_{i} \in \hat{X}^{G}_{s}} k(\hat{x}_{i}, z_{i}) - \frac{2}{n_t}\sum_{\hat{x}_{i} \in \hat{X}^{G}_{t}} k(\hat{x}_{i}, z_{i}) \\
        &+ \frac{\lambda}{2} \|\theta(z_{i})\|_{2}^{2}
    \end{split}
\end{equation}

A large negative value of $L(\theta(z_{i}))$ indicates $z_{i}$ belongs to the target class, and a large positive value of $L(\theta(z_{i}))$ indicates $z_{i}$ belongs to the source class. Therefore, the value of $\theta(z_{i})$ is obtain by minimizing the loss function $L(\theta(z_{i}))$ gradually, until $label(z_{i})$ is $t$. To solve this optimization problem, we use a gradient descent approach:
\begin{equation} \label{eq:gd}
    \begin{split}
        \bigtriangledown_{\theta(z_{i})} L(\theta(z_{i}))=&\frac{1}{n_s}\sum_{\hat{x}_{i} \in \hat{X}^{G}_{s}} k(\hat{x}_{i}, z_{i}) \frac{\hat{x}_i-z_{i}}{\sigma^2} \\
        &- \frac{1}{n_t}\sum_{\hat{x}_{i} \in \hat{X}^{G}_{t}} k(\hat{x}_{i}, z_{i})\frac{\hat{x}_i-z_{i}}{\sigma^2} \\
        &+ \lambda \cdot \theta(z_{i})
    \end{split}
\end{equation}

\begin{equation} \label{eq:gd}
    \begin{split}
        \theta(z_{i})=\theta(z_{i})-\alpha \bigtriangledown_{\theta(z_{i})} L(\theta(z_{i}))
    \end{split}
\end{equation}

where we use RBF kernel $k(x_i, x_j)=e^{-\frac{\|x_i-x_j\|_2^2}{2\sigma^2}}$ as an example, and $\alpha$ is the learning rate. Finding the most appropriate kernel function is beyond the scope of this paper, and there are a few papers discussing about kernel selection \cite{jebara2004multi, kim2006optimal}. In the experimental evaluation, we use the kernel function that gives the highest cross-validation accuracy on the training data.

{\bf \textit{Privacy Guarantee.}} Considering the ``two-class'' scenario described in Section~\ref{sec:ProblemStatement}, we assume the adversaries' approach $A(z, X_{P_1}, X_{P_2})$ would be certain kernel-based classification models trained by public available dataset $[X_{P_1}, X_{P_2}]$. Inspired by semantic security \cite{goldreich2009foundations}, we give our definition of a privacy-preserving data releasing framework as below.

\begin{definition} {\bf (Privacy-preserving Data Releasing Framework.)} $F$ is a privacy-preserving data releasing framework, if given predefined privacy target and certain adversaries' approach $A$, the advantage $Adv[A, F] = |Pr(s = P_1) - Pr(s = P_2)|$ is negligible. (It is straightforward to generate this definition to multi-class scenarios.)
\end{definition}

\begin{theorem}
Our proposed framework is a privacy-preserving data releasing framework.
\end{theorem}

\begin{proof}
Given predefined privacy target, certain appropriate kernel function and public available dataset $[X_{P_1}, X_{P_2}]$, our framework precisely perturbs the private data $x$ towards a perturbed data $z$ associated with a randomly selected privacy target label $t$. Then, given $z$ and $[X_{P_1}, X_{P_2}]$, we have $Pr(s = P_1)=Pr(s = P_2)=\frac{1}{2}$. Therefore, $Adv[A, F] = |Pr(s = P_1) - Pr(s = P_2)|=0$ is negligible.
\end{proof}

\section{Experimental Evaluation} \label{sec:ExperimentalEvaluation}
\subsection{Datasets} \label{sec:ExperimentalEvaluationDatasets}
We have tested our proposed frame with three public datasets: Human Activity Recognition (HAR) \cite{anguita2013public}, Census Income (Census) \cite{kohavi1996scaling} and Bank Marketing (Bank) \cite{moro2014data}. Each dataset has been split into three subsets: training samples (for perturbation approaches), testing samples (data owner's private data), and adversary training samples (for inference attacks).

{\bf Human Activity Recognition (HAR) \cite{anguita2013public}} \textit{HAR} contains smartphone sensor data (i.e., accelerometer data) of 30 subjects' daily activities, where each sample has 561 features and two labels: activities of daily living (ADL) and identity (ID).
In our experiments, we consider ADL as the utility target and ID as the privacy target.
Specifically, ADL has 6 types of labels: ``Walking'', ``Walking Upstairs'', ``Walking Downstairs'', ``Sitting'', ``Standing'' and ``Laying''. On the other hand, ID has 30 types of labels, since 30 subjects have contributed to this dataset. The original dataset is unbalanced. For instance, some subjects contribute more data than the others and some ADLs happen more often than the others. As such, for each different ADL-ID combination ($6 \times 30 = 180$ combinations in total), we randomly sampled 20 samples from the original dataset, resulting in 3,600 samples. The numbers of training, testing and adversary training samples are 1,440, 720 and 1,440, respectively. We kept the number of samples in all ADL-ID combinations equal in all sets.

{\bf Census Income (Census) \cite{kohavi1996scaling}} \textit{Census} has been used to predict whether someone's income exceeds \$50K/yr based on census data. We identify two labels of this dataset: ``income'', where the data user tries to classify if someone's income is ``high'' (higher than \$50K/yr) or ``low'' (lower or equal to \$50K/yr); and ``gender'' (i.e., male/female) which was one feature in the original dataset. Since based on the application, either ``income'' or ``gender'' can be served as utility or privacy targets, we experimented for both cases.
Firstly, we removed the samples with missing features. Secondly, we turned all categorical features into numerical features using binary encoding, which resulted in 51 features. Lastly, we randomly sampled 750 samples for each income-gender combination ($2 \times 2 = 4$ combinations in total) from the original dataset, resulting in 3,000 samples. The numbers of training, testing and adversary training samples are 1,200, 600 and 1,200, respectively. As with the \textit{HAR} dataset, we kept the number of samples in all income-gender combinations equal in all sets.

{\bf Bank Marketing (Bank) \cite{moro2014data}} \textit{Bank} is related with direct marketing campaigns (phone calls) of a Portuguese banking institution. The original classification goal is to predict if the client will subscribe a term deposit (marketing purpose). As such, we used the marketing purpose (``marketing'') as the utility target, which is a ``yes'' or ``no'' binary classification problem. We used marital status (``marital'') as the privacy target, which was one feature in the original dataset. Since very few samples have ``unknown'' marital status, we removed those samples. Thus, ``marital'' has 3 types of labels: ``divorced'', ``married'' and ``single''. As with the \textit{Census}, we turned all categorical features into numerical features using binary encoding, resulting in 31 features. We randomly sampled 410 samples for each marketing-marital combination ($2 \times 3 = 6$ combinations in total) from the original dataset, resulting in 2,460 samples. The numbers of training, testing and adversary training samples are 984, 492 and 984, respectively. We also kept the number of samples in all marketing-marital combinations equal in all sets.

\subsection{Setups}
We evaluate the performance of our proposed two perturbation approaches step-by-step, in terms of utility and privacy. In all the experiments, we utilized RBF-kernel SVM to train the machine learning classifiers for both the utility and privacy targets. The utility classifier is to provide certain premised valuable service, while the privacy classifier is to perform the adversaries' inference attack. All the experiments were performed 15 iterations. At each iteration, a 10-fold cross-validation grid search was performed to find the best set of parameters for training utility and privacy classifiers. As discussed in the last section (Section~\ref{sec:ExperimentalEvaluationDatasets}), we evaluate our frame using three datasets and four scenarios. Given a scenario and its dataset, the evaluation metric is the accuracy of its utility/privacy classifiers. Higher accuracy of the utility classifier means providing better utility. Lower accuracy of the privacy classifier means less privacy leakage. The baseline (i.e., lowest accuracy, no privacy leakage) of the privacy classifier should be equal to the probability of random-guess, of which the prediction is drawn i.i.d. from a uniform distribution.

For the coarse-grained perturbation, we compare our proposed JUPA with a full-dimensional baseline method and four existing dimensionality reduction methods, including Random Projection, PCA, DCA and MDR. Moreover, We evaluate JUPA with regularization parameter $\rho_{0}=0.001$, and different combinations of privacy-utility adjustment parameters $\rho_{1}=1,$ $10^{2},$ $10^{4}$, $\rho^{\prime}_{1}=1,$ $10^{2},$ $10^{4}$. For the fine-grained perturbation, we set $\lambda=0.001$, $\alpha=0.1$, and use a zero vector to initiate the starting noise vector $\theta(z_{0})$.

\subsection{Experimental Results}
Table~\ref{table:HAR}, Table~\ref{table:Census}, Table~\ref{table:Census2} and Table~\ref{table:Bank} shows the experimental results of four scenarios (three datasets), and the following are the main observations and conclusions drawn from experimental results.

\subsubsection{} Considering the coarse-grained perturbation approach alone, JUPA outperforms the other DR methods in terms of the utility and ``somewhat privacy''.
Compared with PCA and random projection, DCA, MDR and JUPA provide better balance between the utility and ``somewhat privacy'' performance, since PCA and random projection are not leveraging any help from the ``label'' information.
For instance, in Table~\ref{table:HAR}, after applying random projection (coarse-grained), the accuracy of ID (privacy) dropped from 62.78\% to 13.75\% (providing one of the best privacy performance), and the accuracy of ADL (utility) dropped from 97.22\% to 60.28\% (giving one of the worst utility performance). On the contrary, after applying PCA (coarse-grained), the accuracy of either utility or privacy does not drop much (providing less ``somewhat privacy'').
Compared with DCA and MDR, JUPA provides better utility and ``somewhat privacy'' performance under certain privacy parameters. For instance, in Table~\ref{table:HAR}, when $\rho_{1}=1$ and $\rho^{\prime}_{1}=1$, compared with other DR methods, JUPA (coarse-grained) provides the highest accuracy (96.11\%) of ADL (utility), and also the second lowest accuracy (only higher than random projection) (21.11\%) of ADL (utility). Results in the other scenarios are inline with this observation.

\subsubsection{} JUPA provides the flexibility for finding a favorable trade-off or tuning between utility and privacy by tuning the privacy parameters. Based on our results, by increasing $\rho_{1}=1$ or $\rho^{\prime}_{1}=1$, JUPA weights more emphasis on preserving privacy (providing accuracy) with small amount of accuracy drop on the utility. For instance, in Table~\ref{table:HAR}, adjusting JUPA from $\rho_{1}=1$, $\rho^{\prime}_{1}=1$ to $\rho_{1}=10^4$, $\rho^{\prime}_{1}=10^4$, results in a 42.78\% drop of the ID (privacy) accuracy (from 21.11\% to 12.08\%), while only resulting in a 8.96\% drop of the ADL (utility) accuracy (from 96.11\% to 87.50\%).

\subsubsection{} Our fine-grained perturbation approach provides the privacy guarantee. For instance, in all the scenarios, after applying the fine-grained perturbation, the accuracies of privacy targets are all converging to or near to the probability of random-guess.

\subsubsection{} In our framework, combining JUPA with the fine-grained perturbation outperforms the other options in terms of the utility. For instance, in Table~\ref{table:HAR}, compared with other DR methods, DCA, MDR and JUPA (fine-grained) provide relative higher accuracy of ADL (utility) ($\leq$ 86.11\%), and when $\rho_{1}=1$ and $\rho^{\prime}_{1}=1$, JUPA provides the best utility accuracy (94.31\%). Results in the other scenarios are inline with this observation. The reason is that even though the fine-grained perturbation could provide guarantee for privacy, applying supervised DR methods (DCA, MDR and JUPA) reserves more utility information and need less noise added to the coarse-grained perturbed data to achieve the privacy guarantee.




\begin{table}[!t]
\footnotesize
\captionsetup{font=footnotesize}
\centering\caption{The Mean Accuracy Percentage Results of Human Activity Recognition Dataset with $K=5$, ADL being the utility target, and ID being the privacy target.}
\label{table:HAR}
\centering
\begin{tabular}{l|c|c|c|c}
\hline
  \multirow{2}{*}{Projection Method} & \multicolumn{2}{c|}{ADL} & \multicolumn{2}{c}{ID}\\\cline{2-5}
    & Coarse & Fine & Coarse & Fine \\
  \hline
  Full-Dimensional & 97.22 & 66.94 & 62.78 & 3.33 \\
  \hline
  Random Projection & 60.28 & 57.36 & 13.75 & 3.33 \\
  \hline
  PCA & 84.72 & 73.33 & 30.28 & 3.75 \\
  \hline
  DCA & 94.58 & 93.75 & 23.61 & 3.33 \\
  \hline
  MDR & 91.67 & 88.75 & 22.92 & 4.58 \\
  \hline
  JUPA ($\rho_{1}=1$, $\rho^{\prime}_{1}=1$) & 96.11 & 94.31 & 21.11 & 3.75 \\
  \hline
  JUPA ($\rho_{1}=1$, $\rho^{\prime}_{1}=10^2$) & 95.83 & 93.47 & 20.28 & 3.61 \\
  \hline
  JUPA ($\rho_{1}=1$, $\rho^{\prime}_{1}=10^4$) & 95.56 & 93.47 & 19.72 & 3.33 \\
  \hline
  JUPA ($\rho_{1}=10^2$, $\rho^{\prime}_{1}=1$) & 94.44 & 93.33 & 20.00 & 3.33 \\
  \hline
  JUPA ($\rho_{1}=10^2$, $\rho^{\prime}_{1}=10^2$) & 94.17 & 92.64 & 17.78 & 3.33 \\
  \hline
  JUPA ($\rho_{1}=10^2$, $\rho^{\prime}_{1}=10^4$) & 93.75 & 92.36 & 16.67 & 3.33 \\
  \hline
  JUPA ($\rho_{1}=10^4$, $\rho^{\prime}_{1}=1$) & 92.50 & 88.19 & 13.61 & 3.33 \\
  \hline
  JUPA ($\rho_{1}=10^4$, $\rho^{\prime}_{1}=10^2$) & 89.58 & 86.39 & 12.50 & 3.33 \\
  \hline
  JUPA ($\rho_{1}=10^4$, $\rho^{\prime}_{1}=10^4$) & 87.50 & 86.11 & 12.08 & 3.33 \\
\hline
\end{tabular}
\end{table}

\begin{table}[!t]
\footnotesize
\captionsetup{font=footnotesize}
\centering\caption{The Mean Accuracy Percentage Results of Census Income Dataset with $K=1$, income being the utility target, and gender being the privacy target.}
\label{table:Census}
\centering
\begin{tabular}{l|c|c|c|c}
\hline
  \multirow{2}{*}{Projection Method} & \multicolumn{2}{c|}{income} & \multicolumn{2}{c}{gender}\\\cline{2-5}
    & Coarse & Fine & Coarse & Fine \\
  \hline
  Full-Dimensional & 84.50 & 69.76 & 87.33 & 50.00 \\
  \hline
  Random Projection & 58.33 & 50.50 & 59.17 & 50.00 \\
  \hline
  PCA & 73.33 & 70.33 & 81.67 & 50.00 \\
  \hline
  DCA & 80.00 & 73.50 & 56.00 & 50.00 \\
  \hline
  MDR & 76.67 & 68.33 & 58.00 & 50.00 \\
  \hline
  JUPA ($\rho_{1}=1$, $\rho^{\prime}_{1}=1$) & 82.50 & 75.33 & 55.50 & 50.00 \\
  \hline
  JUPA ($\rho_{1}=1$, $\rho^{\prime}_{1}=10^2$) & 80.00 & 75.16 & 54.67 & 50.00 \\
  \hline
  JUPA ($\rho_{1}=1$, $\rho^{\prime}_{1}=10^4$) & 78.33 & 74.33 & 54.67 & 50.00 \\
  \hline
  JUPA ($\rho_{1}=10^2$, $\rho^{\prime}_{1}=1$) & 79.17 & 74.66 & 55.00 & 50.00 \\
  \hline
  JUPA ($\rho_{1}=10^2$, $\rho^{\prime}_{1}=10^2$) & 77.50 & 74.00 & 54.50 & 50.00 \\
  \hline
  JUPA ($\rho_{1}=10^2$, $\rho^{\prime}_{1}=10^4$) & 76.67 & 73.83 & 54.17 & 50.00 \\
  \hline
  JUPA ($\rho_{1}=10^4$, $\rho^{\prime}_{1}=1$) & 76.00 & 73.67 & 53.17 & 50.00 \\
  \hline
  JUPA ($\rho_{1}=10^4$, $\rho^{\prime}_{1}=10^2$) & 75.00 & 73.50 & 52.67 & 50.00 \\
  \hline
  JUPA ($\rho_{1}=10^4$, $\rho^{\prime}_{1}=10^4$) & 72.00 & 66.83 & 51.17 & 50.00 \\
\hline
\end{tabular}
\end{table}

\begin{table}[!t]
\footnotesize
\captionsetup{font=footnotesize}
\centering\caption{The Mean Accuracy Percentage Results of Census Income Dataset with $K=1$, gender being the utility target, and income being the privacy target.}
\label{table:Census2}
\centering
\begin{tabular}{l|c|c|c|c}
\hline
  \multirow{2}{*}{Projection Method} & \multicolumn{2}{c|}{gender} & \multicolumn{2}{c}{income}\\\cline{2-5}
    & Coarse & Fine & Coarse & Fine \\
  \hline
  Full-Dimensional & 87.33 & 73.50 & 84.50 & 50.00 \\
  \hline
  Random Projection & 59.17 & 59.17 & 58.33 & 50.00 \\
  \hline
  PCA & 81.67 & 70.33 & 73.33 & 50.00 \\
  \hline
  DCA & 87.50 & 80.50 & 53.17 & 50.00 \\
  \hline
  MDR & 86.67 & 77.83 & 56.00 & 50.00 \\
  \hline
  JUPA ($\rho_{1}=1$, $\rho^{\prime}_{1}=1$) & 88.00 & 82.5 & 57.17 & 50.00 \\
  \hline
  JUPA ($\rho_{1}=1$, $\rho^{\prime}_{1}=10^2$) & 87.67 & 82.17 & 55.67 & 50.00 \\
  \hline
  JUPA ($\rho_{1}=1$, $\rho^{\prime}_{1}=10^4$) & 87.50 & 82.17 & 55.50 & 50.00 \\
  \hline
  JUPA ($\rho_{1}=10^2$, $\rho^{\prime}_{1}=1$) & 87.67 & 81.33 & 55.67 & 50.00 \\
  \hline
  JUPA ($\rho_{1}=10^2$, $\rho^{\prime}_{1}=10^2$) & 86.67 & 81.17 & 54.67 & 50.00 \\
  \hline
  JUPA ($\rho_{1}=10^2$, $\rho^{\prime}_{1}=10^4$) & 86.00 & 80.17 & 54.67 & 50.00 \\
  \hline
  JUPA ($\rho_{1}=10^4$, $\rho^{\prime}_{1}=1$) & 87.00 & 80.33 & 54.33 & 50.00 \\
  \hline
  JUPA ($\rho_{1}=10^4$, $\rho^{\prime}_{1}=10^2$) & 86.67 & 79.67 & 53.50 & 50.00 \\
  \hline
  JUPA ($\rho_{1}=10^4$, $\rho^{\prime}_{1}=10^4$) & 85.67 & 78.67 & 52.67 & 50.00 \\
\hline
\end{tabular}
\end{table}

\begin{table}[!t]
\footnotesize
\captionsetup{font=footnotesize}
\centering\caption{The Mean Accuracy Percentage Results of Bank Marketing Dataset with $K=1$, marketing being the utility target, and marital being the privacy target.}
\label{table:Bank}
\centering
\begin{tabular}{l|c|c|c|c}
\hline
  \multirow{2}{*}{Projection Method} & \multicolumn{2}{c|}{marketing} & \multicolumn{2}{c}{marital}\\\cline{2-5}
    & Coarse & Fine & Coarse & Fine \\
  \hline
  Full-Dimensional & 86.38 & 69.11 & 45.73 & 34.15 \\
  \hline
  Random Projection & 60.57 & 54.88 & 39.23 & 33.33 \\
  \hline
  PCA & 71.14 & 70.73 & 41.06 & 33.33 \\
  \hline
  DCA & 84.76 & 78.66 & 38.01 & 33.33 \\
  \hline
  MDR & 71.75 & 67.48 & 36.79 & 33.33 \\
  \hline
  JUPA ($\rho_{1}=1.0$, $\rho^{\prime}_{1}=1.0$) & 86.38 & 81.30 & 39.63 & 33.33 \\
  \hline
  JUPA ($\rho_{1}=1.0$, $\rho^{\prime}_{1}=10^2$) & 86.18 & 79.67 & 38.82 & 33.33 \\
  \hline
  JUPA ($\rho_{1}=1.0$, $\rho^{\prime}_{1}=10^4$) & 85.37 & 78.66 & 38.41 & 33.33 \\
  \hline
  JUPA ($\rho_{1}=10^2$, $\rho^{\prime}_{1}=1.0$) & 86.18 & 76.22 & 38.01 & 33.33 \\
  \hline
  JUPA ($\rho_{1}=10^2$, $\rho^{\prime}_{1}=10^2$) & 85.98 & 75.61 & 37.60 & 33.33 \\
  \hline
  JUPA ($\rho_{1}=10^2$, $\rho^{\prime}_{1}=10^4$) & 85.98 & 75.41 & 36.18 & 33.33 \\
  \hline
  JUPA ($\rho_{1}=10^4$, $\rho^{\prime}_{1}=1.0$) & 85.37 & 75.20 & 36.99 & 33.33 \\
  \hline
  JUPA ($\rho_{1}=10^4$, $\rho^{\prime}_{1}=10^2$) & 84.35 & 75.00 & 36.59 & 33.33 \\
  \hline
  JUPA ($\rho_{1}=10^4$, $\rho^{\prime}_{1}=10^4$) & 83.13 & 74.59 & 35.77 & 33.33 \\
\hline
\end{tabular}
\end{table}

\section{Conclusion} \label{sec:Conclusion}
In this paper, we proposed a two-step perturbation-based utility-aware privacy-preserving data releasing framework. In the first step, we proposed JUPA, a supervised DR method, which outperforms several existing DR methods in terms of providing utility and ``somewhat privacy'', and provides the flexibility for finding a favorable trade-off or tuning between utility and privacy by tuning the privacy parameters. In the second step, we proposed a fine-grained perturbation approach, which guarantees to provide the protection against inference attacks on certain predefined privacy targets. In the experimental evaluation, we deployed our frame in four scenarios using three public dataset. The experiment results are inline with our expectations and demonstrating the effectiveness and practicality of our framework. Future work will include and extension of JUPA to support non-linear sub-space projections, and an optimized kernel selection method for our fine-grained perturbation approach.


%
%
\section*{Acknowledgment}
This material is based on research sponsored by the DARPA Brandeis Program under agreement number N66001-15-C-4068.\footnote[1]{The views, opinions, and/or findings contained in this article/presentation are those of the author/presenter and should not be interpreted as representing the official views or policies, either expressed or implied, of the Defense Advanced Research Projects Agency or the Department of Defense.}

\ifCLASSOPTIONcaptionsoff
  \newpage
\fi

\bibliographystyle{IEEEtran}
\bibliography{IEEEabrv,tifs19_ppdr}
\vskip 0pt plus -1fil
\begin{IEEEbiography}[{\includegraphics[width=1.05in,height=2.0in,clip,keepaspectratio]{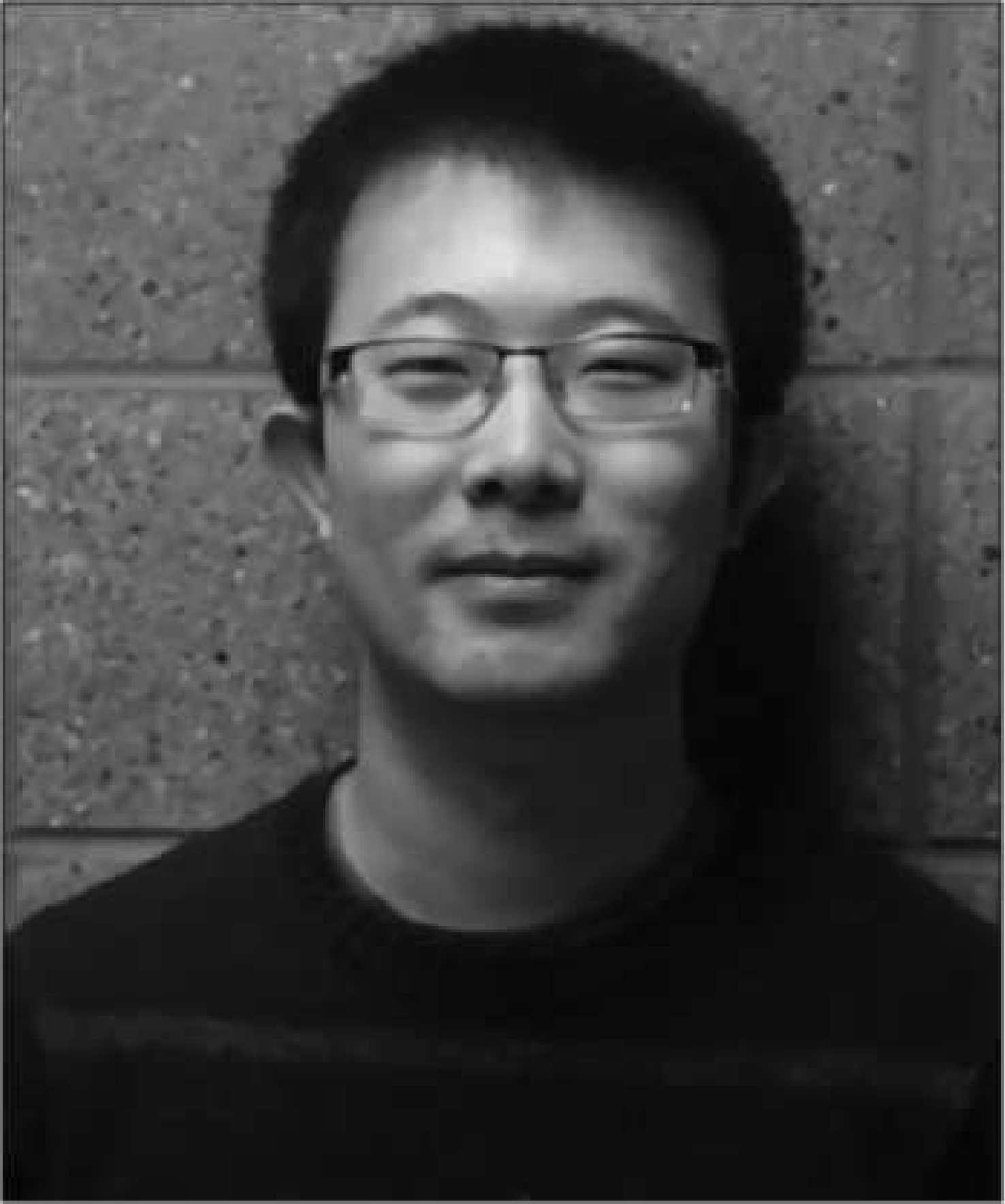}}]{Di Zhuang}
(S'15) received the B.E. degree in computer science and information security from Nankai University, China. He is currently pursuing the Ph.D. degree in electrical engineering with University of South Florida, Tampa. His research interests include cyber security, social network science, privacy enhancing technologies, machine learning and big data analytics. He is a student member of IEEE.
\end{IEEEbiography}
\vskip 0pt plus -1fil
\begin{IEEEbiography}[{\includegraphics[width=1.05in,height=2.0in,clip,keepaspectratio]{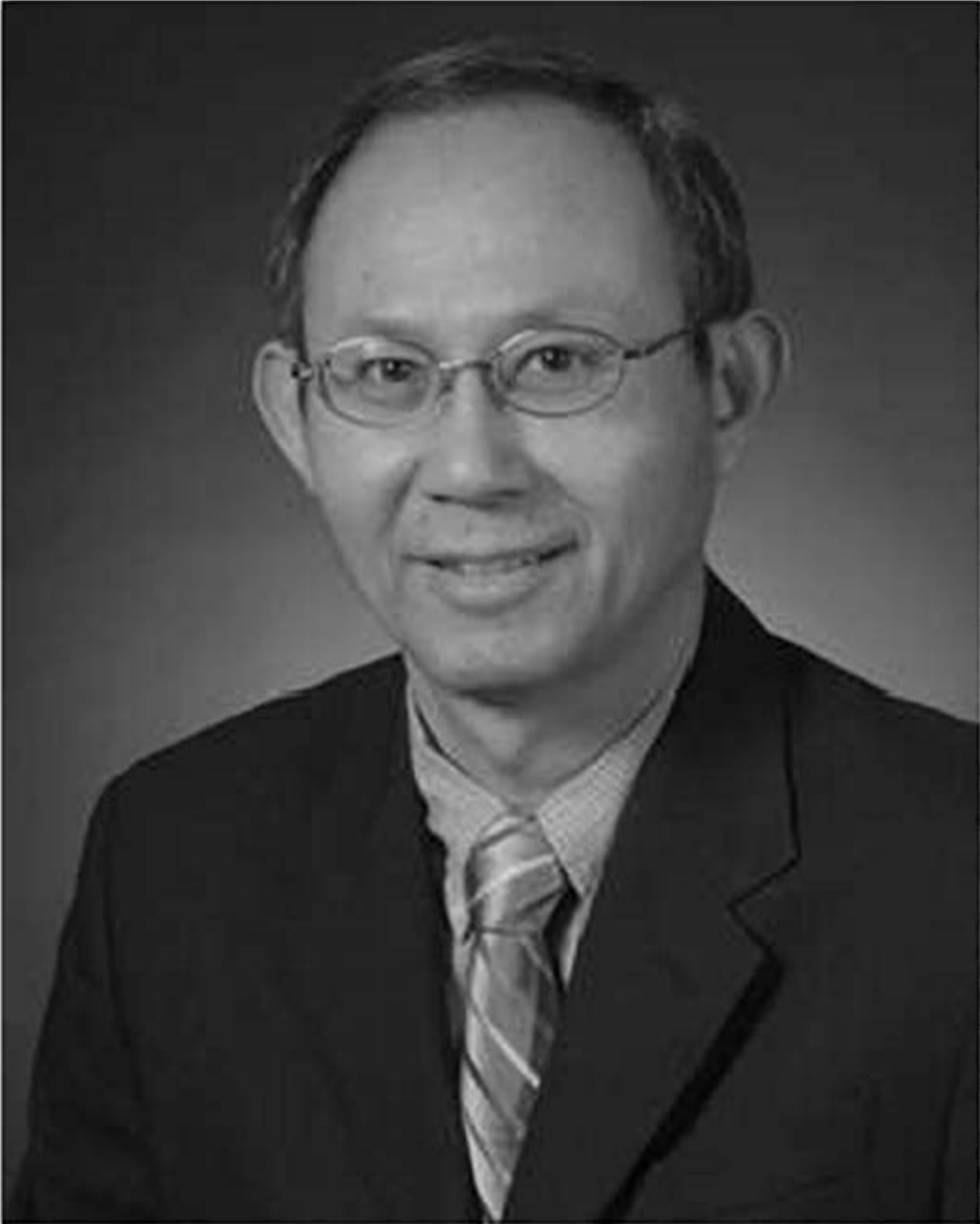}}]{J. Morris Chang}
(SM'08) is a professor in the Department of Electrical Engineering at the University of South Florida. He received the Ph.D. degree from the North Carolina State University. His past industrial experiences include positions at Texas Instruments, Microelectronic Center of North Carolina and AT\&T Bell Labs. He received the University Excellence in Teaching Award at Illinois Institute of Technology in 1999. His research interests include: cyber security, wireless networks, and energy efficient computer systems. In the last six years, his research projects on cyber security have been funded by DARPA. Currently, he is leading a DARPA project under Brandeis program focusing on privacy-preserving computation over Internet. He is a handling editor of Journal of Microprocessors and Microsystems and an editor of IEEE IT Professional. He is a senior member of IEEE.
\end{IEEEbiography}
\end{document}